\documentclass[11pt]{amsart}
\usepackage{amsmath,amssymb,amsfonts,amscd,epsfig}
\usepackage{caption}

\newcommand{\CC}{{\mathbb C}}

\newcommand{\GG}{{\mathbb G}}

\newcommand{\PP}{{\mathbb P}}
\newcommand{\QQ}{{\mathbb Q}}

\newcommand{\ZZ}{{\mathbb Z}}
\newcommand{\ii}{\mathrm{i}}
\newcommand{\dR}{\mathfrak{dr}}
\newcommand{\dd}{\mathrm{d}}
\font \rus= wncyr10
\newcommand{\sha}{\, \hbox{\rus x} \,}

\theoremstyle{plain}
\newtheorem{thm}{Theorem}
\newtheorem{lem}[thm]{Lemma}
\newtheorem{con}[thm]{Conjecture}

\newtheorem{prop}[thm]{Proposition}

\newtheorem{remark}[thm]{Remark}
\newtheorem{defn}[thm]{Definition}
\newtheorem{ex}[thm]{Example}

\title{The Galois coaction on the electron anomalous magnetic moment}
\author{Oliver Schnetz}
\address{Department Mathematik\\
Cauerstra{\ss}e 11\\
91058 Erlangen, Germany}
\email{schnetz@mi.uni-erlangen.de}

\begin{document}
\begin{abstract}
Recently S. Laporta published a partial result on the fourth order QED contribution to the electron anomalous magnetic moment $g-2$ \cite{Laporta}.
This result contains explicit polylogarithmic parts with fourth and sixth roots of unity. In this note we convert
Laporta's result into the motivic `$f$ alphabet'. This provides a much shorter expression which makes the Galois structure visible.
We conjecture the $\QQ$ vector spaces of Galois conjugates of the QED $g-2$ up to weight four.
The conversion into the $f$ alphabet relies on a conjecture by D. Broadhurst that iterated integrals in
certain Lyndon words provide an algebra basis for the extension of multiple zeta values (MZVs) by sixth roots of unity.
We prove this conjecture in the motivic setup.
\end{abstract}

\maketitle
\section{Introduction}
In \cite{Laporta} S. Laporta published a partial result (and a high precision numerical evaluation) of the fourth order QED contribution to the electron anomalous magnetic moment.
The result consists of a polylogarithmic part, an elliptic part, and a yet unknown part. The polylogarithmic part is explicitly given in terms of iterated integrals \cite{Chen}.
The letters in the iterated integrals are from two alphabets with fourth and sixth roots of unity, respectively. We analyze the motivic Galois structure of Laporta's result.

Classical Galois theory associates a finite group to a polynomial $p\in\ZZ[x]$. The elements of the Galois group act as permutation on the roots of $p$ while maintaining the structure
of the field obtained by adding the roots of $p$ to $\QQ$. Because roots of polynomials in one variable are points, the classical Galois group acts on zero-dimensional
objects. In the second half of the 20th century A. Grothendieck envisioned a generalization of Galois theory to higher dimensional varieties. While the idea was (and still is) developed by
many outstanding mathematicians it turned out that dualizing the concept of a group action is both easier to handle and more powerful.

The natural concept of dualizing varieties is the integral. In the algebraic context---which is necessary for a Galois theory to exist---the domain of integration is a variety
(in general with boundary) over $\QQ$, i.e.\ it is given by polynomial inequalities with integer coefficients. Likewise the integrand needs to be a rational function with integer coefficients.
(One may replace the integers by algebraic numbers in $\CC$ without changing the concept.) The variety associated to the differential form is the zero locus of the denominator.
In the case that the result is a mere number it is called a 'period' in \cite{KZ}. Periods form a $\QQ$ algebra which is denoted by $\PP$. In general,
one also wants to allow the polynomials to have variables which are not integrated. In this case the result will be a function on these variables. Here, however,
we can restrict ourselves to periods.

Clearly, the integral dualizes these varieties: It is a map from the domain and the differential form into the complex numbers. By dualizing the Galois action becomes a coaction,
\begin{equation}\label{co}
\Delta\colon \PP \longrightarrow \PP^\dR \otimes_{\QQ} \PP.
\end{equation}
The dual of the Galois group is the Hopf algebra $\PP^\dR$ of deRham periods. One may equivalently swap the sides of the tensor product and co-act with $\PP^\dR$
to the right.

It should be mentioned that strictly speaking the Galois coaction
is an object in algebraic geometry. It needs a sophisticated mathematical theory (the theory of 'motives'). One main conjecture in the theory is that it works identically with
the mere numbers in $\PP$. The conjecture can only fail if there exist some unknown identities in $\PP$ which have no algebraic origin. Still, any formula we obtain
from using the Galois coaction would be valid in $\PP$ because all algebraic relations in the motive are true in $\PP$. Here, we do not make the distinction between numbers and their
motivic counterparts. One other subtlety is that there exist different versions of the Galois coaction. Here we use the simplest one, the unipotent coaction. In the context
of polylogarithms (the 'mixed Tate' case) the unipotent coaction is equivalent to the full coaction.

The concept of classical Galois conjugates is best translated into $\QQ$ vector spaces. In our notation with the left coaction we have periods on the right hand side of the tensor product.
Without giving a precise definition we consider the $\QQ$ vector space of theses periods as Galois conjugates. The simplest cases are
$$
\Delta 1=1^\dR\otimes1,\quad\text{and}\quad\Delta \pi=1^\dR\otimes\pi.
$$
So, the vector space of Galois conjugates of 1 is $\QQ$ while $\QQ\pi$ is conjugate to $\pi$.\footnote{
The idea that multiples of $\pi$ are Galois conjugates of $\pi$ can be motivated by an analogy to classical Galois theory: Multiples of $\pi$ are zeros of a power series
with rational coefficients: $\prod_{n=1}^\infty(1-(\frac{x}{n\pi})^2)=\frac{\sin(x)}{x}\in \QQ[[x]]$. Strictly speaking, zero is not a Galois conjugate of $\pi$.} From
$$
\Delta\log(a)=\log(a)^\dR\otimes1+1^\dR\otimes\log(a),\quad 1\neq a\in\QQ
$$
we obtain that a logarithm has a two-dimensional vector space of Galois conjugates: $\QQ+\QQ\log(a)$.
One finds more explicit examples e.g.\ in references \cite{BrownDecom,coact} where the term coaction is often used without referring to Galois.

A difficulty handling polylogarithmic numbers is that they obey many relations over $\QQ$ (i.e.\ with rational coefficients).
In the case of multiple zeta values (MZVs) it is conjectured that these relations are exhausted by the 'generalized double shuffle relations'. Using these relations
is still difficult at high weights because it is impossible to---a priory---give a sequence of equations that leads to the reduction of a specific MZV to a $\QQ$ basis
(see e.g.\ \cite{Abl,datamine} for this approach). For other types of polylogarithmic numbers there do not even conjecturally exist complete sets of equations.

Thankfully another deep mathematical input resolves this problem. For many polylogarithmic numbers there exist an (at least) conjectural isomorphism into an algebra
where the only relations are shuffles. This shuffle algebra is composed of words in letters of certain weights. For MZVs one has one letter for each odd weight $\geq3$. It is easy to check that
words in these letters span weight spaces of the desired (conjectural) dimensions. Another benefit of this so called '$f$ alphabet' is that the coaction becomes simple deconcatenation,
$$
\Delta w=\sum_{w=uv}u\otimes v.
$$
This, in particular, makes it trivial to re-construct the original word from a known term in the coaction.
Moreover, it is clear that the Galois conjugates of a word $w$ are spanned by its right factors ($v$ if $w=uv$). The general strategy to handle polylogarithmic numbers it to
convert them into an $f$ alphabet. The method for this conversion is the decomposition algorithm by F. Brown \cite{BrownDecom}.

The only grain of salt is that the conversion is non-canonical in the sense that it depends on a choice of an algebra basis for the numbers considered. In many cases such a basis is not
even conjecturally known. For all (known) polylogarithmic numbers in quantum field theory (QFT), however, we have such bases and the conversion is possible.
In general, also non polylogarithmic numbers have $f$ alphabets (see e.g.\ \cite{BrownMotPer}).
Note that the vector spaces of Galois conjugates are canonical although the conversion into the $f$ alphabet is not.

In this note we convert the polylogarithmic part of the QED contribution to $g-2$ into the $f$ alphabet. This conversion simplifies the result. The $\QQ$ vector spaces of Galois conjugates
happen to be of intriguingly small dimensions.

It is proved for generic kinematics in \cite{Brownamplitudes,Browncoaction} that the $\QQ$ vector spaces of Galois conjugates of fixed weights are finite dimensional in QFT. It is conjectured and mostly proved that at weight $n$ all Galois conjugates come from sub-quotient graphs with at most $2n+1$ edges. (Note that the Hodge weight used in \cite{Brownamplitudes,Browncoaction}
is twice the transcendental weight for MZVs.) Our data suggest that the space of weight $n$ Galois conjugates is exhausted by the $n$-loop contribution
to the electron $g-2$. The small dimensions of the $\QQ$ vector spaces are further evidence for a coaction structure in QED similar to the one conjectured in $\phi^4$
theory \cite{coact}.

In this note we mostly process a result by S. Laporta. We try to keep this note short to emphasize that the results should be mostly attributed to S. Laporta.
\vskip1ex

\noindent
For the extension MZV(4) of MZVs by fourth roots of unity (see Section \ref{sect2} for details) P. Deligne gave a basis in 2010 \cite{Deligne}. Its restriction to Lyndon words
provides an algebra basis of MZV(4). In this note we conveniently use a `parity' version of Deligne's basis \cite{coact}.
In its $f$ alphabet one has one letter at each positive integer weight (see Table \ref{fs}),
\begin{equation}
f_n^4\cong\left\{\begin{array}{ll}
2\mathrm{Re}\,\mathrm{Li}_n(\ii),&\text{if $n$ is odd,}\\
2\ii\mathrm{Im}\,\mathrm{Li}_n(\ii),&\text{if $n$ is even.}\end{array}\right.
\end{equation}
We use the sign $\cong$ to emphasize that, strictly speaking, the right hand side is not equal to the left hand side. It is obtained by applying the composition of a non-canonical
(motivic $f$ alphabet) isomorphism with the conjectural period isomorphism to the left hand side.

In the case of sixth roots of unity P. Deligne only gave a basis in \cite{Deligne} for the subspace of MZV(6) that is constructed by using the three letters $0,1,\xi_6$, where
$$
\xi_N=\exp(2\pi\ii/N)
$$
is a primitive $N$th root of unity. For the conversion of the fourth order contribution to $g-2$ this is not sufficient.
In general, we need the full alphabet of $N+1$ letters (for $N=6$)
\begin{equation}
X_N=\{0,\;\xi_N^k, k=0,\ldots,N-1\}.
\end{equation}
In 2015, D. Broadhurst defined a set of iterated integrals with words in $X_6$ (see Definition \ref{basis6}) and conjectured that this set provides an algebra basis for MZV(6) \cite{DB6}.
For mere numbers it seems impossible to prove Broadhurst's conjecture. However, in the (conjecturally equivalent) motivic setup one can use the Galois coaction
to give a (rather simple) proof.

\begin{thm}\label{mainthm}
The motivic Broadhurst set of iterated integrals in sixth roots of unity is an algebra basis for motivic MZV(6).
\end{thm}

The $f$ alphabet for MZV(6) has two generators for weight one and one generator for every weight greater than one.
Here, we again use a parity version of Broadhurst's basis. The letters are
\begin{eqnarray}\label{fn}
f_n^6&\cong&\left\{\begin{array}{ll}
-2\mathrm{Re}\,\mathrm{Li}_n(\xi_3),&\text{if $n$ is odd}\\
-2\ii\mathrm{Im}\,\mathrm{Li}_n(\xi_3),&\text{if $n$ is even}\end{array}\right.\quad n\geq1\text{, and}\\
g^6_1&\cong&2\log2.
\end{eqnarray}

Note that in Laporta's result the two alphabets MZV(4) and MZV(6) do not mix. This makes it possible to avoid the use of twelfth roots of unity.
Also note that (see Prop.\ \ref{prop1})

\begin{equation}
g^6_1\cong-2f^4_1\quad\text{and}\quad\frac{3^{n-1}}{3^{n-1}-1}f^6_n\cong-\frac{4^{n-1}}{2^{n-1}-1}f^4_n\quad\text{for odd $n\geq3$.}
\end{equation}
Although these letters are related they must be distinguished because they belong to different alphabets. However, it is the benefit of the parity basis
that words in $\{f^4_1,f^4_3,f^4_5,\ldots\}$ and $\{g^6_1,f^6_3,f^6_5,\ldots\}$ can (non-trivially) be converted into each other. We have, e.g., the identity
\begin{equation}
g^6_1f^6_3\cong\frac{256}{27}f^4_1f^4_3-\frac{139}{28350}\pi^4.
\end{equation}

The result for the QED contribution $a_e$ to the electron $g-2$ is
\begin{eqnarray}\label{ae}
a_e&\cong&\frac{1}{2}\left(\frac{\alpha}{\pi}\right)\nonumber\\
&&+\,\left(\frac{197}{144}+\frac{1}{12}\pi^2+\frac{27}{32}f^6_3-\frac{1}{4}g^6_1\pi^2\right)\left(\frac{\alpha}{\pi}\right)^2\nonumber\\
&&+\,\left(\frac{28259}{5184}+\frac{17101}{810}\pi^2+\frac{139}{16}f^6_3-\frac{149}{9}g^6_1\pi^2-\frac{525}{32}g^6_1f^6_3+\frac{1969}{8640}\pi^4-\frac{1161}{128}f^6_5\right.\nonumber\\
&&\qquad+\,\left.\frac{83}{64}f^6_3\pi^2\right)\left(\frac{\alpha}{\pi}\right)^3\nonumber\\
&&+\,\left(\frac{1243127611}{130636800}+\frac{30180451}{155520}\pi^2-\frac{255842141}{2419200}f^6_3-\frac{8873}{36}g^6_1\pi^2+\frac{126909}{2560}\frac{f^6_4}{\ii\sqrt{3}}\right.\nonumber\\
&&\qquad-\,\frac{84679}{1280}g^6_1f^6_3+\frac{169703}{3840}\frac{f^6_2\pi^2}{\ii\sqrt{3}}+\frac{779}{108}g^6_1g^6_1\pi^2+\frac{112537679}{3110400}\pi^4-\frac{2284263}{25600}f^6_5\nonumber\\
&&\qquad+\,\frac{8449}{96}g^6_1g^6_1f^6_3-\frac{12720907}{345600}f^6_3\pi^2-\frac{231919}{97200}g^6_1\pi^4+\frac{150371}{256}\frac{f^6_6}{\ii\sqrt{3}}+\frac{313131}{1280}g^6_1f^6_5\nonumber\\
&&\qquad-\,\frac{121383}{1280}f^6_2f^6_4-\frac{14662107}{51200}f^6_3f^6_3
+\frac{8645}{128}\frac{f^6_2g^6_1f^6_3}{\ii\sqrt{3}}-\frac{231}{4}g^6_1g^6_1g^6_1f^6_3-\frac{16025}{48}\frac{f^6_4\pi^2}{\ii\sqrt{3}}\nonumber\\
&&\qquad+\,\frac{4403}{384}g^6_1f^6_3\pi^2-\frac{136781}{1920}f^6_2f^6_2\pi^2+\frac{7069}{75}f^4_2f^4_2\pi^2-\frac{1061123}{14400}f^6_3g^6_1\pi^2\nonumber\\
&&\qquad+\,\frac{1115}{72}\frac{f^6_2g^6_1g^6_1\pi^2}{\ii\sqrt{3}}+\frac{781181}{20736}\frac{f^6_2\pi^4}{\ii\sqrt{3}}-\frac{4049}{1080}g^6_1g^6_1\pi^4+\frac{90514741}{54432000}\pi^6\nonumber\\
&&\qquad-\,\frac{95624828289}{2050048}f^6_7-\frac{29295}{512}g^6_1f^6_2f^6_4+\frac{107919}{512}g^6_1f^6_3f^6_3+\frac{337365}{256}f^6_3g^6_1f^6_3\nonumber\\
&&\qquad-\,\frac{55618247}{409600}f^6_5\pi^2-\frac{1055}{256}g^6_1f^6_2f^6_2\pi^2+\frac{26}{3}f^4_1f^4_2f^4_2\pi^2+\frac{553}{4}g^6_1f^6_3g^6_1\pi^2\nonumber\\
&&\qquad-\,\frac{35189}{1024}f^6_3g^6_1g^6_1\pi^2+\frac{79147091}{2211840}f^6_3\pi^4-\frac{3678803}{4354560}g^6_1\pi^6\nonumber\\
&&\qquad+\,\sqrt{3}(E_{4a}+E_{5a}+E_{6a}+E_{7a})+E_{6b}+E_{7b}+U\Bigg)\left(\frac{\alpha}{\pi}\right)^4.
\end{eqnarray}

\begin{center}
\begin{tabular}{l|l|llllllll}
wt.&dim.&words\\\hline
0&1&1\\
1&0&\text{---}\\
2&1&$\pi^2$\\
3&2&$f^6_3$&$g^6_1\pi^2$\\
4&6&$f^6_4$&$g^6_1f^6_3$&$f^6_2\pi^2$&$f^4_2\pi^2$&$g^6_1g^6_1\pi^2$&$\pi^4$\\\hline
5&4&$f^6_5$&$g^6_1g^6_1f^6_3$&$f^6_3\pi^2$&$g^6_1\pi^4$\\
6&15&$f^6_6$&$g^6_1f^6_5$&$f^6_2f^6_4$&$f^6_3f^6_3$&$f^6_2g^6_1f^6_3$&$g^6_1g^6_1g^6_1f^6_3$&$f^6_4\pi^2$\\
&&$g^6_1f^6_3\pi^2$&$f^6_2f^6_2\pi^2$&$f^4_2f^4_2\pi^2$&$f^6_3g^6_1\pi^2$&$f^6_2g^6_1g^6_1\pi^2$&$f^6_2\pi^4$&$g^6_1g^6_1\pi^4$&$\pi^6$\\
7&11&$f^6_7$&$g^6_1f^6_2f^6_4$&$g^6_1f^6_3f^6_3$&$f^6_3g^6_1f^6_3$&$f^6_5\pi^2$&$g^6_1f^6_3g^6_1\pi^2$&$g^6_1f^6_2f^6_2\pi^2$\\
&&$f^4_1f^4_2f^4_2\pi^2$&$f^6_3g^6_1g^6_1\pi^2$&$f^6_3\pi^4$&$g^6_1\pi^6$
\end{tabular}
\end{center}
\captionof{table}{A basis of the Galois conjugates in (\ref{ae}). In general, higher loop orders in $a_e$ will provide new Galois conjugates which are not listed in the table.
However, we expect that all new Galois conjugates have at least weight five. So, up to weight four the above table should be comprehensive.}
\label{words}
\vskip1ex

All calculations were done using the Maple package {\tt HyperlogProcedures} \cite{hyperlog_procedures}. Galois conjugates are listed in Table \ref{words}.

\begin{remark}
\begin{enumerate}
\item In (\ref{ae}) the letters $f^6_1\cong\log3$ and $f^4_4$ are absent in all words. In particular the absence of the weight one letter $f^6_1$ drastically reduces
the number of allowed words. Note that the absence of $f^6_1$ as a letter is much stronger than the absence of $f^6_1$ as a Galois conjugate ($g^6_1$ is present as a letter but not as a
Galois conjugate).

The absence of the letter $\log3$ in the parity Broadhurst basis of MZV(6) can be given the following description:
An iterated integral $I(1,w,0)$ with a word $w$ in the three letters $0,-1,\xi_3^2$ (or $\xi_3$) is free of the letter $\log3$ if and only if $w$ is of the following type:
(a) $w$ has no letter $\xi_3^2$ or (b) $w$ has a single letter $\xi_3^2$ and no letter $-1$ and weight $\geq2$ or (c) $w$ does neither begin or end in $\xi_3^2$ and the deletion of
all letters 0 in $w$ leads to one of the three structures $\xi_3^2-1\ldots-1$, $-1\ldots-1\xi_3^2$, or $\xi_3^2-1\ldots-1\xi_3^2$. In particular, $w$ can have at most two letters $\xi_3^2$.
Regretfully, these words are not stable under shuffle nor does the restriction of the Broadhurst basis to $\log3$ free elements give an algebra basis for $\log3$ free MZV(6).
It would be desirable to get a more natural description of $\log3$ free MZV(6).
\item There are no odd powers of $\pi$ in (\ref{ae}).
\item At odd weights all Galois conjugates have an odd number of letters (odd coradical depth, see Section \ref{sect2}).
\item The structure of numbers in (\ref{ae}) (i.e.\ MZV($N$) for $N=1,2,4,6$) equals the structure of numbers found (or expected) in primitive graphs of massless $\phi^4$ theory.
However, in massless $\phi^4$ theory these numbers are found at much higher loop orders (loop orders $3,9,8,7$ for $N=1,2,4,6$, respectively).

Note that in $\phi^4$ theory the loop order for first Euler sums, $N=2$, is artificially high by a subtle mechanism connected to the coaction conjectures \cite{coact}.
Naively one would expect Euler sums at loop order 7 in $\phi^4$).

The connection between massive and massless QFTs may hint to the existence of a universal structure which severely restricts the type
of numbers in any QFT calculation (see \cite{mod} for a comparison between $\phi^4$ and unphysical non-$\phi^4$ numbers). By time of writing there exists nothing concrete
about this potential structure (other than that it should be related to the Galois coaction).
\item The coaction principle works graph by graph, so it is not surprising to see it in a partial result. It may be possible, however, that the picture changes slightly
when a full result is known.
\item Even weight letters in $f^6$ need to have a factor of $1/\sqrt3\ii$ because the letters are imaginary but $a_e$ is real. Because all integrals of $a_e$
are defined over $\QQ$ imaginary numbers come as $\sqrt{-3}$ in MZV(6).
\item The elliptic and the unknown contributions $E$ and $U$ also have Galois conjugates.
\item We expect that the $\QQ$ vector space of Galois conjugates of weight $n$ is completely determined by loop orders $\leq n$. In particular, we expect that in Table \ref{words} the list
of polylogarithmic (mixed Tate) Galois conjugates up to weight four is complete to all loop orders in $a_e$.
\item The rich structure in (\ref{ae}) supports the validity of Laporta's result.
\end{enumerate}
\end{remark}

\subsection*{Acknowledgements}
The author is grateful for very helpful discussions with F. Brown, C. Glanois, and E. Panzer on motivic MZVs. The author is supported by the DFG grant SCHN 1240/2-2.

\section{Proof of theorem \ref{mainthm}}\label{sect2}
\subsection{Iterated integrals}
Iterated integrals were defined on manifolds by Chen \cite{Chen}. Here, we only fix the notation in the elementary case of a punctured sphere. Let $a_1,\ldots,a_n,z\in\CC$, then
the iterated integral $I(z,a_n\ldots,a_1,0)$ is recursively defined by $I(z,0)=1$ and 
\begin{eqnarray}
\frac{\partial}{\partial z}I(z,a_na_{n-1}\ldots a_1,0)&=&\frac{1}{z-a_n}I(z,a_{n-1}\ldots a_1,0),\nonumber\\
I(0,a_na_{n-1}\ldots a_1,0)&=&0\quad\text{if }n\geq1.
\end{eqnarray}
Iterated integrals with fixed endpoints $0,z$ form a shuffle algebra in the words $a_n\ldots a_1$,
\begin{equation}\label{shuffle}
I(z,u,0)I(z,v,0)=I(z,u\sha v,0),
\end{equation}
where iterated integrals are extended to sums of words by linearity.

Iterated integrals converge (are regular) if and only if $a_1\neq0$ and $a_n\neq z$. The definition of iterated integrals can be extended to all words by the following procedure:
For $z\neq0$ we set
\begin{equation}\label{reg}
I(z,0^{\{n\}},0)=(-1)^nI(z,z^{\{n\}},0)=\log(|z|)^n/n!,
\end{equation}
where  $a^{\{n\}}$ is a sequence of $n$ letters $a$.
The absolute value in the logarithm originates from a tangential base point regularization prescription applied to straight paths.
We obtain the following 'un-shuffle' formula:
\begin{lem}\label{unshufflelemma}
For any letters $a,b$ and any word $w$ we have
\begin{equation}\label{unshuffle}
I(z,wba^{\{n\}},0)=\sum_{i=0}^n(-1)^iI(z,a^{\{n-i\}},0)I(z,[w\sha a^{\{i\}}]b,0).
\end{equation}
\end{lem}
\begin{proof}
Straight forward induction over $n$ using the identity
\begin{equation}
I(z,a,0)I(z,wba^{\{n-1\}},0)=nI(z,wba^{\{n\}},0)+I(z,[w\sha a]ba^{\{n-1\}},0).
\end{equation}
\end{proof}
We use (\ref{unshuffle}) and the analogous formula for reversed words to express non-convergent (singular) iterated integrals as polynomials in $\log(|z|)$ with coefficients
which are sums of regular iterated integrals. If $|z|=1$---which we have in the context of this article---only the constant term $i=n$ survives.
In the proof of Theorem \ref{mainthm} we will need the fact that this constant term is an integer linear combination of regular iterated integrals.

For completeness we give a formula for the simultaneous left-right un-shuffle. Let $a\neq b$, then
$$
I(z,b^{\{m\}}wa^{\{n\}},0)=\sum_{i=0}^n\sum_{j=0}^m(-1)^{i+j}I(z,b^{\{m-j\}},0)I(z,a^{\{n-i\}},0)I(z,[b^{\{j\}}\sha w\sha a^{\{i\}}]_{\rm reg},0),
$$
where the subscript 'reg' means that we keep only those words that neither begin in $a$ nor end in $b$. The proof is by induction over $m$. We do not need this statement here and hence
leave the details to the reader.

It is important to notice that the extension to singular iterated integrals is consistent with the shuffle product (\ref{shuffle}).
A more detailed account on iterated integrals in this specific context can e.g.\ be found in \cite{coact,gf} and the references therein.

The $\QQ$ algebra $I$ of (motivic) iterated integrals has a coaction. The unipotent part of this coaction can be considered as a map
\begin{equation}
\Delta I\longrightarrow (I\mod2\pi\ii)\otimes I.
\end{equation}
An explicit formula for $\Delta$ co-acting on $I$ was given by A. Goncharov and F. Brown \cite{Gon}, \cite{MMZ}, \cite{P3P}.

\subsection{Extensions of multiple zeta values}
By geometric expansion a regular iterated integral can be expressed as a sum. For $a_i\in\CC\setminus\{0\}$, $i=1,\ldots,r$ we have
\begin{equation}\label{sum}
(-1)^rI(z,0^{\{n_r-1\}}a_r\ldots0^{\{n_1-1\}}a_1,0)=\sum_{k_r>\ldots>k_1\geq1}\frac{(\frac{z}{a_r})^{k_r}\ldots(\frac{a_2}{a_1})^{k_1}}{n_r^{k_r}\ldots n_1^{k_1}}
=\mathrm{Li}_{n_r,\ldots,n_1}\left(\frac{z}{a_r},\ldots,\frac{a_2}{a_1}\right).
\end{equation}
With (\ref{unshuffle}) singular iterated integrals are polynomials in $\log(|z|)$ with coefficients given by sums of Lis.

In the special case $z=1$ and $a_i\in X_N\setminus\{0\}$, $i=1,\ldots,r$, we obtain extensions MZV($N$) of MZVs by $N$th roots of unity.

The motivic structure of MZV($N$) is conjecturally understood if $N$ is not a prime power $p^n$ for $p\geq5$ \cite{Zhao,DG}. In these cases there exists
an (in most cases conjectural) isomorphism $\psi$ into an $f$ alphabet where multiplication is shuffle and the coaction is deconcatenation.
The isomorphism $\psi$ is non-canonical; it depends on a choice of an algebra basis of MZV($N$).

\begin{center}
\begin{tabular}{l|lll}
$N$&weight 1&even weight&odd weight $\geq3$\\\hline
1&0&0&1\\
2&1&0&1\\
$\neq 1,2,p^n$ for $p\geq5$&$\phi(N)/2+\nu(N)-1$&$\phi(N)/2$&$\phi(N)/2$
\end{tabular}
\end{center}
\captionof{table}{Number of letters in the $f$ alphabet of motivic MZV($N$), where $\phi$ is Euler's totient function and $\nu(N)$ is the number of distinct primes in $N$.
The last line in this table is mostly conjectural.}
\label{fs}
\vskip1ex

The number of letters in the $f$ alphabet depends on $N$ (see Table \ref{fs}). The case $N=1$ was solved by F. Brown \cite{MMZ} (with a contribution by D. Zagier
\cite{Zagier23}). In the case $N=2,3,4,8$ a vector space basis of MZV($N$) was given by P. Deligne in \cite{Deligne}. This basis can trivially be transformed into an
algebra basis by restriction to Lyndon words. In this section we address the case $N=6$.

In \cite{Deligne} P. Deligne also considers the case $N=6$. However, for $N=6$ he restricts the alphabet in the iterated integrals to the three letters $0,1,\xi_6$.
In 2015, D. Broadhurst provided a conjectural algebra basis for the full MZV(6) (Definition \ref{basis6}) \cite{DB6} (see also \cite{HSS6}).
Theorem \ref{mainthm} proves that this basis is an algebra basis for motivic MZV(6). For other $N$ some partial results were obtained by C. Glanois (private communication).

\subsection{Sixth roots of unity and the Broadhurst conjecture}
In \cite{DB6} D. Broadhurst defines the following sets of numbers:
\begin{defn}\label{basis6}
Consider the three numbers $0,\xi_3^2=(-1-\sqrt3\ii)/2,\xi_2=-1$ as letters and order them by $0\prec\xi_3^2\prec\xi_2$. Then
\begin{eqnarray}
B^6_1&=&\{2\pi\ii,I(1,\xi_3^2,0),I(1,\xi_2,0)\},\nonumber\\
B^6_n&=&\{I(1,w,0),w\hbox{ Lyndon word in }0,\xi_3^2,\xi_2\hbox{ with no $0\xi_2$ subsequence}\}.
\end{eqnarray}
\end{defn}

\begin{ex}\label{ex1}
\begin{eqnarray}
B^6_2&=&\{I(1,0\xi_3^2,0),I(1,\xi_3^2\xi_2,0)\}\nonumber\\
B^6_3&=&\{I(1,00\xi_3^2,0),I(1,0\xi_3^2\xi_3^2,0),I(1,0\xi_3^2\xi_2,0),I(1,\xi_3^2\xi_3^2\xi_2,0),I(1,\xi_3^2\xi_2\xi_2,0)\}.
\end{eqnarray}
\end{ex}

\begin{remark}
Because MZV(6) is stable under complex conjugation one could equally well use the letter $\xi_3$ instead of $\xi_3^2$.
For the parity version of the Broadhurst basis (which we use for the QED $g-2$) complex conjugation merely flips the signs in all words with odd parity \cite{coact}.
In the following we decided to choose the letter $\xi_3^2$ for the following three (minor) reasons: (1) Broadhurst uses the differential form $-\dd(1-\xi_3 x)=-\dd x/(x-\xi_3^2)$
in his original work \cite{DB6}. This corresponds to the letter $\xi_3^2$. (2) With the letter $\xi_3^2$ words with no letter $-1$ are iterated integrals in $0,1$ with endpoint $\xi_3$
which is also used in Deligne's basis for MZV(3). (3) The author's Maple package \cite{hyperlog_procedures} uses the letter $\xi_3^2$.
\end{remark}

\begin{con}[D. Broadhurst, 2015]
The Broadhurst set $B^6=\{B^6_n,n\geq1\}$ is an algebra basis for MZV(6).
\end{con}

Because transcendentality statements are notoriously hard, it seems impossible to prove the above conjecture. However, in the conjecturally equivalent motivic setup
where the coaction is available the conjecture can be proved (this is Theorem \ref{mainthm}).

For the proof of Theorem \ref{mainthm} we need the following propositions:
\begin{prop}\label{prop1}
Let $\zeta$ be the Riemann zeta function. We have
\begin{eqnarray}\label{prop}
\zeta(n)&\in&\frac{(1+(-1)^{n-1})3^{n-1}}{1-3^{n-1}}\mathrm{Li}_n(\xi_3)+\QQ(2\pi\ii)^n\quad\text{if }n\in\{2,3,4,\ldots\},\nonumber\\
\mathrm{Li}_n(\xi_6)&\in&(2^{1-n}+(-1)^n)\mathrm{Li}_n(\xi_3)+\QQ(2\pi\ii)^n\quad\text{if }n\in\{1,2,3,4,\ldots\}.
\end{eqnarray}
\end{prop}
\begin{proof}
This follows from Lemma 3.1 and Lemma 3.2 in \cite{coact}.
\end{proof}
The depth of a word is the number of its non-zero letters.
\begin{prop}\label{prop2}
Let $w$ be a word of weight $n$ and depth one and let $a,b$ letters, all in the alphabet $0,\xi_3^2,\xi_2$. Then
\begin{equation}\label{zero}
I(b,w,a)\in c\mathrm{Li}_n(\xi_3)+\QQ(2\pi\ii)^n,\quad\text{with }\QQ\ni c\equiv0\mod3.
\end{equation}
\end{prop}
\begin{proof}
We write $I(b,w,a)\equiv0\mod2\pi\ii\mod3$ for (\ref{zero}). From (\ref{prop}) we have
$$
\zeta(n)\equiv0\mod2\pi\ii\mod3\quad\text{for }n\geq2,\quad\quad\mathrm{Li}_n(\xi_6)\equiv0\mod2\pi\ii\mod3\quad\text{for }n\geq1.
$$
By Lemma 3.1 in \cite{coact} complex conjugation on $\mathrm{Li}_n(\xi_6)$ gives a factor of $(-1)^{n-1}$ modulo $2\pi\ii$. Therefore we also have
$\mathrm{Li}_n(\xi_6^5)\equiv0\mod2\pi\ii\mod3\text{ for }n\geq1$.

If $w$ does not end in 0 we have four types of iterated integrals that begin in 0:
\begin{equation}\label{n1}
I(\xi_2,0^{\{n-1\}}\xi_2,0)=I(\xi_3^2,0^{\{n-1\}}\xi_3^2,0)=\left\{\begin{array}{cl}-\zeta(n)&\text{if }n>1\\
0&\text{if }n=1 \text{ (by Eq.\ (\ref{reg}))}\end{array}\right.
\end{equation}
and
\begin{equation}\label{n2}
I(\xi_2,0^{\{n-1\}}\xi_3^2,0)=-\mathrm{Li}_n(\xi_6^5),\quad I(\xi_3^2,0^{\{n-1\}}\xi_2,0)=-\mathrm{Li}_n(\xi_6).
\end{equation}
These iterated integrals are 0 modulo $2\pi\ii$ and modulo 3. Because $\log(|b|)=0$ Lemma \ref{unshufflelemma} extends this to all iterated integrals $I(b,w,0)$.
By path concatenation and path reversal for iterated integrals (see e.g.\ \cite{gf}) we have
\begin{equation}\label{pathcon}
I(b,w,a)\equiv I(b,w,0)+(-1)^nI(a,\widetilde{w},0)\mod\text{products,}
\end{equation}
where $\widetilde{w}$ is $w$ in reversed order. The words in the missing product terms have total depth one, implying that one factor has depth zero.
By path concatenation, path reversal, and (\ref{reg}) this factor is zero. So, we can use (\ref{pathcon}) to complete the proof of the proposition.
\end{proof}

\subsection{Proof of Theorem \ref{mainthm}}
We use the ordered letters $g_1\cong\log2\succ f_1\cong-\mathrm{Li}_1(\xi_3)\succ f_2\cong-\mathrm{Li}_2(\xi_3)\succ f_3\cong-\mathrm{Li}_3(\xi_3)\succ\ldots$
for the $f$ alphabet with respect to the Broadhurst set $B^6$. Let
\begin{equation}
{\mathcal U^6}:=\langle\hbox{words in }g_1,f_n\rangle_\QQ.
\end{equation}
We endow ${\mathcal U^6}$ with the shuffle product and the deconcatenation coproduct. This makes ${\mathcal U^6}$ a Hopf algebra.
The general theory on MZV(6) provides the existence of an isomorphism \cite{BrownMotPer}
\begin{equation}\label{iso}
\psi\,:\,\text{MZV}(6)\longrightarrow {\mathcal U^6}\otimes_\QQ\QQ[2\pi\ii].
\end{equation}
The Hopf algebra ${\mathcal U^6}$ co-acts on  ${\mathcal U^6}\otimes_\QQ\QQ[2\pi\ii]$ by deconcatenation on the left hand side of the tensor product.
The isomorphism $\psi$ can be constructed with the decomposition algorithm by F. Brown \cite{BrownDecom} (see also Example 2.1 in \cite{coact}).

The decomposition algorithm inductively cuts off left letters in ${\mathcal U^6}$ by applying derivatives $\delta_{g_1}$ and $\delta_{f_m}$, $m=1,2,3,\ldots$.
More precisely, the derivatives $\delta_x$ first map into ${\mathcal U^6}$, then cut off the left letter $x$ in ${\mathcal U^6}$, and finally map back into MZV(6).
In this construction words in ${\mathcal U^6}$ which do not begin in $x$ are nullified. Schematically,
$$
\delta_x:\text{MZV(6)}\longrightarrow\text{MZV(6)},\quad\delta_x=\psi^{-1}\circ\delta^{\mathcal U^6}_x\circ\psi.
$$
If the letter $x$ has weight $m$ then one can reduce the calculation of $\delta_x$ on any
iterated integral with weight $n>m$ to the case $n=m$ (see Eq.\ (2.13) in \cite{coact}):
\begin{multline}\label{delta}
\delta_x I(a_{n+1},a_n\ldots a_1,a_0)=\\
\sum_{k=0}^{n-m}(\delta_x I(a_{k+m+1},a_{k+m}\ldots a_{k+1},a_k))I(a_{n+1},a_n\ldots a_{k+m+1}a_k\ldots a_1,a_0).
\end{multline}
Note that $\delta_x I(a_{k+m+1},a_{k+m}\ldots a_{k+1},a_k)\in\QQ$ is a number (because the weight of the iterated integral matches the weight of the letter $x$).
By definition $\delta_x I=0$ if $I$ has weight less than $m$.

\begin{lem}\label{lem1}
Let $v$ be a word in the alphabet $X_6$. Then
\begin{eqnarray}\label{dmod3}
\delta_{f_m}I(1,0^{\{n-1\}}\xi_3^2v,0)&\equiv&\left\{\begin{array}{cl}I(1,v,0)&\text{if }n=m\\
0&\text{if }m<n\\
\ast&\text{if }m>n\end{array}\right.\hspace{.8ex}\mod\text{lower depth}\mod3,\nonumber\\
\delta_{g_1}I(1,av,0)&\equiv&\left\{\begin{array}{cl}I(1,v,0)&\text{if }a=\xi_2\\
0&\text{otherwise}\end{array}\right.\mod\text{lower depth}\mod3.
\end{eqnarray}
The above equations mean that we first nullify all expressions which have lower depth than the words in the iterated integral of the left hand sides.
Then we reduce the coefficients of all iterated integrals in (\ref{delta}) modulo 3 to obtain the right hand sides. The $\ast$ in the first equation merely
means that the reduction modulo 3 exists (no coefficient has a factor of 3 in the denominator).
\end{lem}
\begin{proof}
By definition of the basis $B^6$ we have (in (\ref{fn}) we use the parity basis)
\begin{equation}\label{psii}
\psi(I(1,0^{\{n-1\}}\xi_3^2,0))=-\psi(\mathrm{Li}_n(\xi_3))=f_n.
\end{equation}
Applying this to (\ref{zero}) in Prop.\ \ref{prop2} we get
\begin{equation}\label{psizero}
\psi(I(b,w,a))\equiv0\mod2\pi\ii\mod3
\end{equation}
for any depth one word $w$ and letters $a,b$ in the alphabet $0,\xi_3^2,\xi_2$.

We use (\ref{delta}) on $B^6$. The words in the iterated integrals on the right hand side give a partition of the word $a_n\ldots a_1$ on the left hand side.
If the word $a_{k+m}\ldots a_{k+1}$ has depth two or more then the right hand side vanishes modulo lower depth.

By path concatenation, path reversal, and (\ref{reg}), iterated integrals of depth zero words vanish for endpoints $0,\xi_3^2,\xi_2,1$.
So, we may assume that $a_{k+m}\ldots a_{k+1}$ has depth one.

By (\ref{psizero}) only the term $k=n-m$ contributes modulo $2\pi\ii$ and modulo 3. By (\ref{pathcon}) and (\ref{zero}) we have ($a_{n+1}=1$)
\begin{equation}\label{II}
I(1,a_n\ldots a_{n-m+1},a_{n-m})=I(1,a_n\ldots a_{n-m+1},0)\mod\text{products}\mod2\pi\ii\mod3.
\end{equation}
As in the proof of Prop.\ \ref{prop2} the product terms vanish.

Equations (\ref{delta}) and (\ref{psii}) give the non-zero terms on the right hand sides of (\ref{dmod3}).

If $m<n$ in the first equation or $a=0$ in the second equation of (\ref{dmod3}) then the word $a_n\ldots a_{n-m+1}$ has depth zero.
The corresponding iterated integral vanishes.

If $m>n$ in the first equation then the only possible depth one sequence for $a_n\ldots a_{n-m+1}$ is $u=0^{\{n-1\}}\xi_3^20^{\{m-n\}}$. The corresponding iterated integral
$I(1,u,0)$ in (\ref{II}) un-shuffles over $\ZZ$ (by (\ref{unshuffle})) to $I(1,0^{\{m-1\}}\xi_3^2,0)$ which is mapped to $f_m$ (by (\ref{psii})).
The integer coefficient has a reduction modulo 3.

If $a=\xi_3^2$ in the second equation of (\ref{dmod3}) then $m=1$ and $I(a_{n+1},a_n,0)=I(1,\xi_3^2,0)$. By definition $\psi(I(1,\xi_3^2,0))=f_1$ with zero $g_1$ component.
Hence $\delta_{g_1}(I(1,\xi_3^2,0))=0$ and (\ref{dmod3}) is proved.
\end{proof}

We inductively define a `naive' map $\psi_0$ from words in the alphabet $X_6$ into ${\mathcal U^6}$ which preserves {\em no} structure of MZV(6) but which may be considered as
a conductor for the construction of $\psi$. For words $u,v$ in the alphabet $X_6$ let
\begin{eqnarray}
\psi_0(u\xi_2v)&=&\psi_0(u)g_1\psi_0(v),\nonumber\\
\psi_0(u0^{\{n-1\}}\xi_3^2v)&=&\psi_0(u)f_n\psi_0(v)\quad\text{if the rightmost letter of $u$ is not 0}.
\end{eqnarray}
In the order of Example \ref{ex1} we have $\psi_0($words in $B^6_2)=\{f_2,f_1g_1\}$ and $\psi_0($words in $B^6_3)=\{f_3,f_2f_1,f_2g_1,f_1f_1g_1,f_1g_1g_1\}$.

\begin{lem}\label{oneone}
The map $\psi_0$ is one to one between Lyndon words in $B^6\setminus\{2\pi\ii\}$ and Lyndon words in $f,g$.
\end{lem}
\begin{proof}
A word $w$ is Lyndon if it is smaller than all its non-trivial rotations. Assume $w\neq 0^{\{n\}}$ is a word in $X_6$. Because 0 is the smallest letter in $X_6$,
rotating a zero $w=0v\mapsto v0$ always makes the word larger. To decide if $w$ is Lyndon it therefore suffices to only check rotations of $w$ which do not end in zero.
If $w$ is a word on which $\psi_0$ is defined (i.e.\ $w$ has no subsequence $0\xi_2$ and $w$ does not end in zero) then $\psi_0$ is defined on all its relevant rotations (not ending in zero).

It hence suffices to show that $\psi_0$ preserves the ordering on these words. We do this by induction over the number of letters in the image of $\psi_0$.
Note that the ordering of the letters in $\mathcal U^6$ is defined such that the claim holds on them. Consider two distinct words $au,bv\in \mathcal U^6$ with initial letters $a,b$.
Then $au\succ bv$ is equivalent to $a\succ b$ or $a=b$ and $u\succ v$. By induction, in both cases this is equivalent to $\psi_0^{-1}(au)\succ\psi_0^{-1}(bv)$.
\end{proof}

By Radford's theorem Lyndon words generate the shuffle algebra $\mathcal U^6$.
The above lemma implies that the number of weight $n$ elements in $B^6\setminus\{2\pi\ii\}$ matches the number of weight $n$ generators in $\mathcal U^6$.
In the isomorphism (\ref{iso}) the element $2\pi\ii\in B^6_1$ generates the polynomial algebra $\QQ[2\pi\ii]$.
So, from the existence of the isomorphism (\ref{iso}) it follows that the number of elements in $B^6_n$ matches the number of weight $n$ algebra generators in MZV(6).

With these preparations we can now prove Theorem \ref{mainthm}.
\begin{proof}[Proof of Theorem \ref{mainthm}]
The construction of $\psi$ with the decomposition algorithm can only fail if bijectivity fails at some weight $n$.
Because (\ref{iso}) simply maps $2\pi\ii$ onto $1\otimes 2\pi\ii$ it suffices to show bijectivity modulo $2\pi\ii$.
Because $\psi$ (by construction) is compatible with the $\QQ$ algebra structures it also suffices to show bijectivity modulo products.
With both reductions the isomorphism $\psi$ induces a linear map $\psi'$ from the $\QQ$ span of $B^6\setminus\{2\pi\ii\}$ into the
$\QQ$ span of Lyndon words in ${\mathcal U^6}$.

By construction the weight is preserved by $\psi$ and $\psi'$. So, we have to show that every weight $n$ part $\psi'_n$ of $\psi'$ is a (finite dimensional) bijective linear map.
(In fact with the statement after Lemma \ref{oneone} it suffices to show injectivity or surjectivity of $\psi'_n$.)
The existence of these bijective linear maps $\psi'_n$ is a general feature of the decomposition algorithm.

The coradical depth of an $x\in {\mathcal U^6}$ is the maximum number of letters in any word of $x$.
So, e.g., $2f_2+3f_1g_1$ has coradical depth 2. The decomposition algorithm uses induction over the coradical depth.
From the proof of Lemma \ref{lem1} it is clear that $\psi$ cannot increase the depth.

We use Lemma \ref{lem1} to show that
\begin{equation}\label{mod3eq}
\psi'_n(I(1,w,0))\equiv\psi_0(w)+\sum_{\text{Lyndon }v\prec\psi_0(w)}c_vv \mod\text{lower depth}\mod3
\end{equation}
for any word $w$ of weight $n$ in letters $X_6$ (with coefficients $c_v\in\ZZ/3\ZZ$).
Because the naive map $\psi_0$ preserves the depth, calculating modulo lower depth in Lemma \ref{lem1} is equivalent to calculating modulo lower coradical depth in (\ref{mod3eq}).
So, the non-zero terms on the right hand sides of (\ref{dmod3}) iteratively generate the term $\psi_0(w)$ on the right hand side of (\ref{mod3eq}).

From the existence of Eq.\ (\ref{dmod3}) we see by induction that the right hand side of (\ref{mod3eq}) can be reduced modulo 3 after the reduction modulo lower depth.

Let $y$ be the first letter of $\psi_0(w)$ and let $x\succ y$. Because $g_1$ is the largest letter we have $y=f_n$. Hence, the first letter of $w$ is not $\xi_2$.
In both cases---$x=f_m$ for $m<n$ and $x=g_1$---we have $\delta_xI(1,w,0)=0$ modulo lower depth and modulo 3 from the zero terms on the right hand sides of (\ref{dmod3}). This implies by induction that $\psi'_n(I(1,w,0))$ cannot
produce terms $\succ\psi_0(w)$ modulo lower depth and modulo 3. So, (\ref{mod3eq}) is proved.

It follows from (\ref{mod3eq}) that $\psi'_n$ is block triangular in the depth. By induction over depth it suffices to show the bijectivity of $\psi'_n$ modulo lower depth.

We sort the elements in the depth $r$ part of $B^6_n$ by $\prec$ and likewise sort the Lyndon words of coradical depth $r$ in ${\mathcal U^6}$.
In the proof of Lemma \ref{oneone} we showed that $\psi_0$ preserves the ordering. From (\ref{mod3eq}) it follows that the matrix of $\psi'_n$ modulo lower depth becomes triangular modulo 3
with respect to the sorted bases. On the diagonal all entries are 1 modulo 3. Hence $\psi'_n$ modulo lower depth is invertible modulo 3.
This implies the bijectivity of $\psi'$ and the theorem follows.
\end{proof}
Altogether, in the proof of Theorem \ref{mainthm} we have reduced modulo $2\pi\ii$, products, lower depths, and the number 3.

Note that we have also proved that the map $\psi'_n$ (and hence also the isomorphism $\psi$) preserves the depth on $B^6_n$. In this sense $B^6_n$ is a basis of minimum depth.

Deligne's results for $N=2,3,4,8$ can also be proved with the method used in this section.

\bibliographystyle{plain}
\renewcommand\refname{References}

\end{document}